\documentclass[aps,twocolumn,superscriptaddress,nofootinbib,pra,10pt]{revtex4-1}
\usepackage{amsmath,amsfonts,amssymb,amsthm,graphicx,bbm,enumerate,times}
\usepackage[caption=false]{subfig}
\usepackage{mathtools}

\newcommand{\ket}[1]{\left \vert #1 \right \rangle}
\newcommand{\bra}[1]{\left \langle #1 \right \vert}

\newcommand{\hh}{\mathcal{H}}
\newcommand{\id}{\mathbb{I}}
\newcommand{\cc}{\mathbbm{C}}
\newcommand{\xx}{\mathbf{x}}
\newcommand{\yy}{\mathbf{y}}
\newcommand{\St}{{\mathcal S}}

\DeclareMathOperator{\poly}{poly}
\DeclareMathOperator{\Eq}{EQ}

\newcommand{\ketbra}[2]{\left \vert #1 \right \rangle \! \!\left \langle #2 \right \vert}

\newtheorem{theorem}{Theorem}
\newtheorem{lemma}[theorem]{Lemma}
\newtheorem{corollary}[theorem]{Corollary}

\theoremstyle{definition}
\newtheorem{example}[theorem]{Example}
\newtheorem{definition}[theorem]{Definition}

\DeclareMathOperator{\tr}{tr}
\DeclareMathOperator{\vspan}{span}

\begin{document}

\title{Area laws and efficient descriptions of quantum many-body states}
\author{Yimin Ge}
\affiliation{Max-Planck-Institut  f{\"u}r Quantenoptik, D-85748 Garching, Germany}

\author{Jens Eisert}
\affiliation{Dahlem Center for Complex Quantum Systems, Freie Universit{\"a}t Berlin, D-14195 Berlin, Germany}


\begin{abstract}
It is commonly believed that area laws for entanglement entropies imply that a quantum many-body state can be faithfully represented by efficient tensor network states -- a conjecture frequently stated in the context of numerical simulations and analytical considerations. In this work, we show that this is in general not the case, except in one dimension. We prove that the set of quantum many-body states that satisfy an area law for all Renyi entropies contains a subspace of exponential dimension. 
Establishing a novel link between quantum many-body theory and the theory of communication complexity,
we then show that there are states satisfying area laws for all Renyi entropies but cannot be approximated by states with a classical description of small Kolmogorov complexity, including polynomial projected entangled pair states (PEPS) or states of multi-scale entanglement renormalisation (MERA). Not even a quantum computer with post-selection can efficiently prepare all quantum states fulfilling an area law, and we show that not all area law states can be eigenstates of local Hamiltonians. We also prove translationally invariant and isotropic instances of these results, and show a variation with decaying correlations using quantum error-correcting codes.
\end{abstract}
\maketitle

\section{Introduction} 

 Complex interacting quantum systems show a wealth of 
exciting phenomena, ranging from phase transitions of zero temperature to notions of topological order. A significant proportion of condensed matter physics is 
concerned with understanding the features and properties emergent in quantum lattice systems with local interactions. Naive numerical descriptions of such 
quantum systems with many degrees of freedom require prohibitive resources, however, for the simple
reason that the dimension of the underlying Hilbert space grows exponentially in the system size. 

It has  
become clear in recent years,
however, that ground states -- and a number of other natural states -- 
usually occupy only 
a tiny fraction of this Hilbert space, 
sometimes referred to as its
``physical corner'' {(Fig.~\ref{subfig:corner1})}. This subset is commonly characterised by states having little entanglement by exhibiting an area law \cite{AreaRMP}: 
entanglement entropies are expected to 
grow only like the  boundary area of any subset $A$ of 
lattice sites,
  \begin{equation}
	S(\rho_A) = O(|\partial A|),
\end{equation}
and not extensively like its volume $|A|$ (Fig.~\ref{fig:peps}). Such area laws have been proven for all gapped models in $D=1$ \cite{hastings07,ALV12,AKLV13,BH13}, for free gapped bosonic and fermionic models
in $D>1$ \cite{Area,Area2,Area3}, for ground states of gapped models in the same phase as ones satisfying an area law \cite{vAMV13,MAAV14},
those which have a suitable scaling for heat capacities \cite{HeatCapacity} or for which the Hamiltonian spectra satisfy related conditions \cite{HastingsEntropyEntanglement,MasanesArea},
frustration-free spin models \cite{FF},
and ones that exhibit local topological order \cite{michalakis2012}. 
The general expectation is that all gapped lattice models satisfy such a behaviour -- proving a general area law
for gapped lattice models in $D\geq 2$ has indeed become a milestone open problem in theoretical physics.

This behaviour 
is at the core of powerful numerical algorithms, 
such as the density-matrix renormalisation group approach \cite{schollwoeck2011} 
and higher dimensional analogues \cite{VMC08}. 
In $D=1$,
the situation is particularly clear: 
Matrix-product states essentially ``parameterise'' those one-dimensional quantum states that satisfy an area law for some Renyi entropy $S_\alpha$ with $\alpha\in(0,1)$. They approximate such states provably well, which explains why essentially machine precision can be reached with such numerical tools \cite{SchuchApprox,VC06}. 
Analogously, 
a common  jargon is that higher dimensional analogues -- 
projected entangled pair states (PEPS) -- 
 can approximate states satisfying area laws, 
 for the same reasoning and with analogue implications. In  
the same way, one expects those
instances of tensor network states to capture the ``physical corner''.

\begin{figure}[t!]
\includegraphics[width=0.3\textwidth]{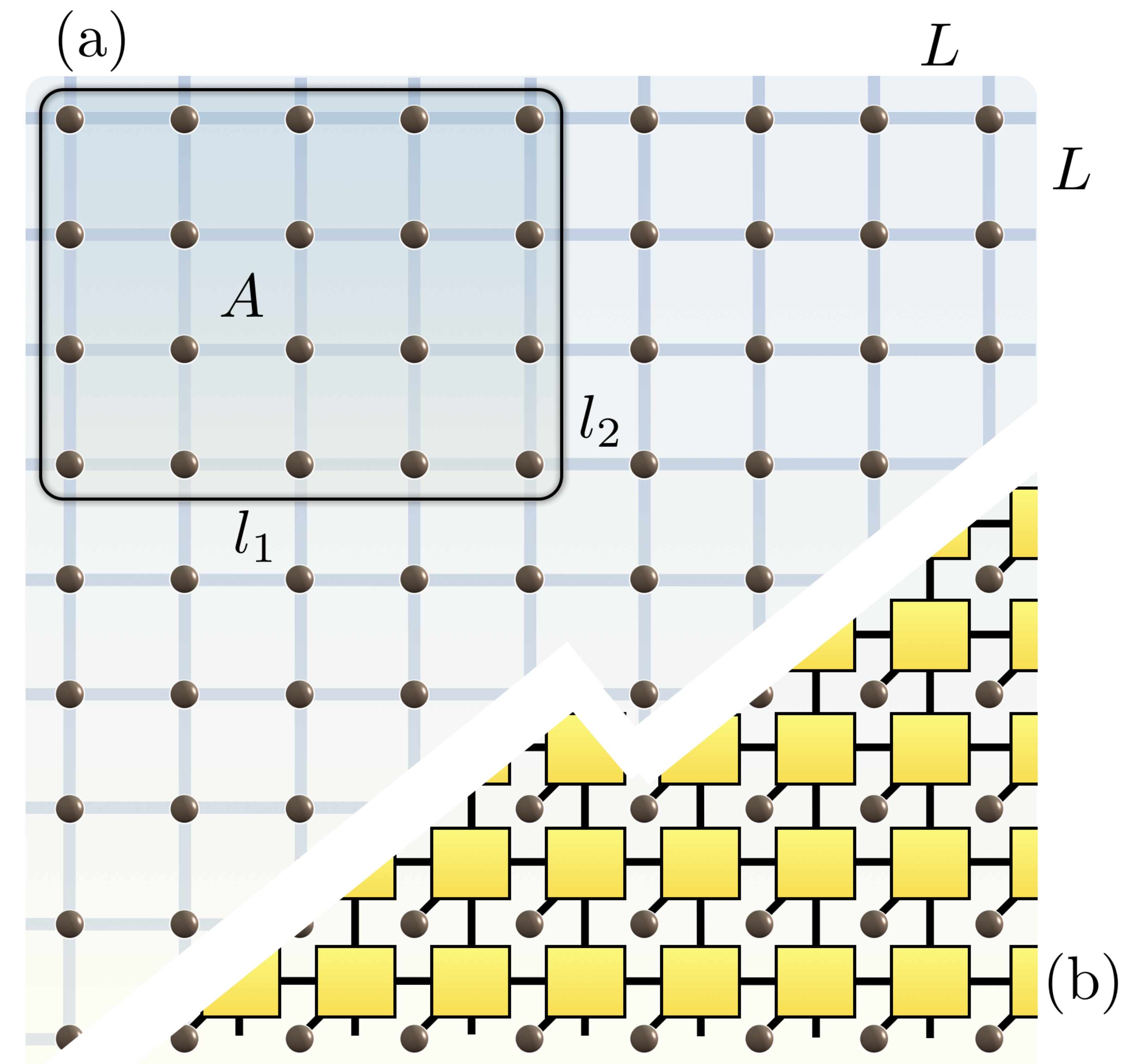}
\caption{(a) There exist quantum states on $D$-dimensional cubic lattices in $D\geq 2$ such that $S_\alpha(\rho_A)=O(|\partial A|)$ for all
$\alpha>0$, but which cannot be approximated by efficient tensor network states, such as (b) polynomial projected entangled pair states.}\label{fig:peps}
\end{figure}

 In this work, we show that this jargon is not quite right: Strictly speaking, area laws and approximability with tensor network states 
 are unrelated. There are even states that satisfy an area
 law for \emph{every} Renyi entropy \footnote{Here, $S_1=\lim_{\alpha\downarrow 1}S_\alpha=S$ is the familiar von-Neumann entropy and $S_0$ the binary logarithm of the Schmidt rank.}
 \begin{equation}
 	S_\alpha = \frac{1}{1-\alpha}\log_2 \tr(\rho^\alpha), \quad \alpha\in[0,\infty),
\end{equation}	
but still, no efficient PEPS can be found. The same holds for multi-scale entanglement renormalisation (MERA) ansatzes \cite{MERA1}, 
as well as
all tensor network states that have a short description  
with
low Kolmogorov complexity. Not even a quantum computer with post-selection
can prepare all states satisfying area laws. Moreover, not all states satisfying area laws are eigenstates of local Hamiltonians.

The main result of this {work}, which underlies these conclusions, is that  in 
$D\geq 2$,
 the set of states satisfying area laws for all $S_\alpha$ contains a subspace whose dimension scales exponentially with the system size. Bringing the study of many-body states and tensor network states into contact with the theory of 
\emph{communication complexity}  
and \emph{Kolmogorov complexity} \cite{Kolmogorov}, it can then be inferred that this subspace cannot be parameterised by polynomial classical descriptions only.

By no means, however, is this result meant to indicate that area laws are not appropriate intuitive guidelines for
approximations with tensor network states.
It is rather aimed to be a significant step towards precisely delineating the boundary between those quantum many-body 
states that can be efficiently captured and those that cannot, and we contribute to the discussion why PEPS and other tensor network states
approximate natural states so well.
Area laws without further qualifiers are, strictly speaking, inappropriate for this purpose as the ``corner'' they parameterise is exponentially large. This work is hence 
 a strong reminder that the programme of identifying that boundary is not 
 finished yet.

\section{Area laws and the exponential ``corner'' of Hilbert space}
Throughout this work, we consider quantum lattice systems of local dimension $d$, arranged on a cubic lattice $[L]^{D}$ of  
dimension  $D>1$, where $[L]:=\{0,\dots, L-1\}$. 
The case $D=1$ is excluded since in this case, the question at hand has already been settled
with the opposite conclusion \cite{SchuchApprox,VC06}. 
The local dimension is small and taken to be $d=3$
for most of this work, 
there is no obvious
fundamental reason, however, why such a
construction should not also be possible for $d=2$.

In the focus of attention are states that satisfy an area law for all Renyi  
entropies, including the  
von Neumann entropy. 

\begin{definition}[Strong area laws] \label{def:arealaw}
 A pure state  $\psi \in \St((\cc^d)^{\otimes L^{D}})$ is said to satisfy a \emph{strong area law} if there exists a universal constant $c$ such that for all  
regions $A\subset [L]^D$, we have
 $S_0(\psi_A)\leq c|\partial A|$. 
\end{definition}

Since $S_\alpha(\rho)\leq S_0(\rho)$ for all $\alpha>0$, 
strong area law states 
in this sense also 
exhibit area laws for all {other} Renyi entropies.  Definition~\ref{def:arealaw} is hence even stronger than the area laws usually quoted \cite{AreaRMP,SchuchApprox,VC06}.
{Here and later, we write $\psi=\ketbra\psi\psi$}.
For simplicity, we will for the remainder of this paper restrict our consideration to cubic regions 
$A$ 
only. It should be clear, however, that all arguments generalise to arbitrary regions $A\subset [L]^D$.

We now turn to {showing} that  the ``physical corner'' of states satisfying area laws in this strong sense is still
very large: It contains subspaces of dimension
$\exp(\Omega(L^{D-1}))$. We prove this by providing a specific class of quantum states that have that property. At the heart of the construction is an embedding
of states defined on a $(D-1)$-dimensional qubit lattice into the $D$-dimensional qutrit one.  
Denote with 
$\hh_L \subset \vspan\{\ket 1,\ket 2\}^{\otimes L^{D-1}}$ the subspace of translationally invariant states on a $(D-1)$-dimensional cubic lattice of $L^{D-1}$ qubits. 
It is easy to show that $\dim (\hh_L)\geq 2^{L^{D-1}}/L^{D-1}$.
We start from the simplest translationally invariant construction on $\hh :=(\cc^3)^{\otimes L^D}$
and discuss isotropic states and decaying correlations below. 

\begin{theorem}[States satisfying strong area laws]\label{thm:areaLaw}
	There exists an injective linear isometry $f:\hh_L\rightarrow \hh$ 
	with the property that  for all  $\ket{\phi_L}\in \hh_L$, $f(\ket{\phi_L})$ satisfies a strong area law and is translationally invariant in all $D$ directions.
\end{theorem}
\begin{figure}[t!]
\includegraphics[width=0.48\textwidth]{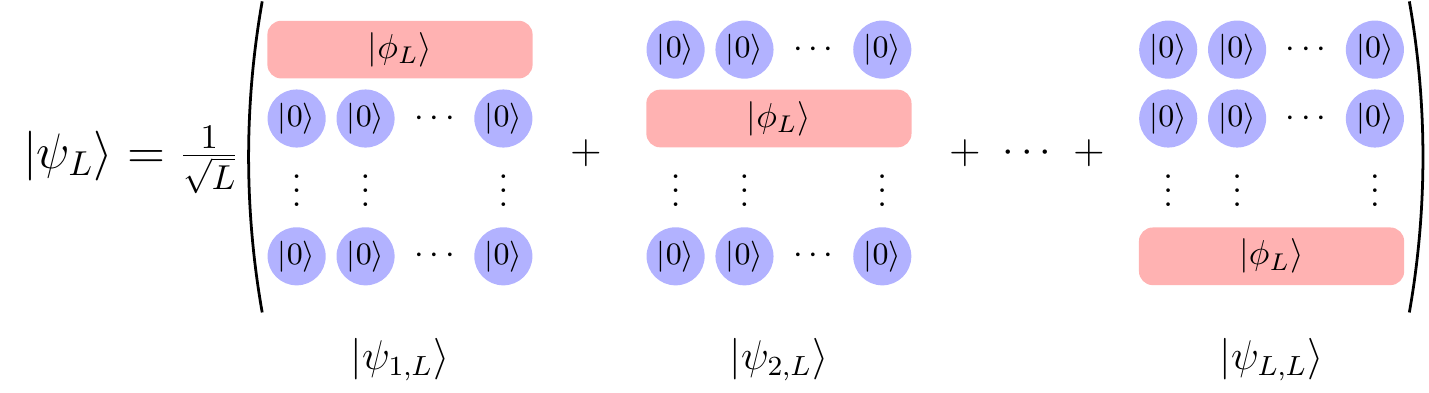}
\caption{Schematic drawing of $\ket{\psi_L}$ in $D=2$. $\ket{\phi_L}$ is an arbitrary translationally invariant state vector on $L^{D-1}$ qubits with basis states $\ket1,\ket2$ in $D-1$ dimensions. Schmidt decompositions for $\ket{\psi_L}$ with respect to bi-partitions of the lattice can be readily obtained from the corresponding Schmidt decompositions of $\ket{\phi_L}$.}
\label{fig:psiLfig}
\end{figure}
\begin{proof}
Given a state vector $\ket{\phi_L}\in\hh_L$, define
\begin{equation}\label{cst}
	|\psi_{k,L}\rangle := |0\rangle^{\otimes (k-1)L^{D-1}}\otimes |\phi_L\rangle \otimes |0\rangle^{\otimes (L-k)L^{D-1}}\in\hh,
\end{equation}
with $\ket{\phi_L}$ at the $k$-th hyperplane of the lattice. 
Define
\begin{equation}\label{fullstate}
	 \ket{\psi_L} := L^{-1/2} \sum_{k=1}^L |\psi_{k,L}\rangle,
\end{equation}
which is translationally invariant (see Fig.~\ref{fig:psiLfig}).
Any such state vector will satisfy a strong area law (in fact, a sub-area law): For any cubic subset $A= [ l_1] \times\cdots\times [l_D]$,
we have
for the reduced state $(\psi_L)_A = \tr_{\bar{A}} (\psi_L)$ that
\begin{align}
	 S_0((\psi_L)_A)&\leq \log_2( 2^{l_1\cdots l_{D-1}}l_D +1)\nonumber\\
	 &\leq 2\sum_{j=1}^D \prod_{k, k\neq j} l_k
	 =|\partial A|,
\end{align}
where we used that the Schmidt rank with respect to the bi-partition $A,\bar A$ for each $\ket{\psi_{k,L}}$ with $ k\in[l_D]$ is at most
$2^{l_1\ldots l_{D-1}}$, and that 
since $\ket{\phi_L}$ is only supported on $\text{span}\{\ket1, \ket2\}$,
the Schmidt vectors of $\ket{\psi_{k,L}}$ and $\ket{\psi_{k',L}}$ are orthogonal for $k\neq k'\in[l_D]$ such that in the distinguished $D$-th direction,
the contribution to the Schmidt rank is additive and thus linear in $l_D$. 
Setting $f(\ket{\phi_L}):=\ket{\psi_L}$, we see that $f$ has the desired properties.
\end{proof}

\section{Classically efficiently described states}
We now turn to efficient classical descriptions of quantum many-body systems. 
The focus is on
 tensor network states, but we will see 
that the notion of an efficient classical description can be formulated in a much more general way. In this {fashion}, we establish a link between 
tensor network states and those quantum states having a small Kolmogorov complexity. 
We then review why the exponential dimension in Theorem~\ref{thm:areaLaw} shows that not all translationally invariant  strong area law states can be approximated by states with a  polynomial classical description. 
For our purposes, the following definition of efficiently describable quantum states will suffice (see also Ref.~\cite{Kolmogorov} for {alternative} definitions). 

\begin{definition}[Classical descriptions]\label{def:classical}
	A \emph{classical description} of a pure quantum state {$\psi\in\St((\cc^d)^{\otimes N})$} 
	is a Turing machine that outputs the list of the coefficients of $\ket\psi$ in the standard basis $\{\ket{\xx}: \xx\in [d]^N\}$ and halts. 
	The \emph{length} of the classical description is the size of the Turing machine. 
	We say that the description is \emph{polynomial} if its length is polynomial in $N$. 
\end{definition}
We emphasise that for a polynomial classical description we only require the \emph{size} of the Turing machine to be polynomial, but not the \emph{run-time} (which is necessarily exponential). 

\begin{example}
[Tensor networks]
\label{ex:tensor-network}
States 	
	that can be written as polynomial tensor networks, i.e., defined on arbitrary graphs with bounded {degree}, having at most $O(\poly(N))$ bond-dimension and 
	whose tensor entries have at most $O(\poly(N))$ Kolmogorov complexity\footnote{Recall that the Kolmogorov complexity of a classical string $w$ is the size of the shortest 
	Turing machine that outputs $w$ and halts. It can be thought of as the shortest possible (classical) description of $w$.} are polynomially classically described states in the sense of Definition~\ref{def:classical}. In particular, PEPS and MERA states with $O(\poly(N))$ bond-dimension and tensor entries of at most $O(\poly(N))$ Kolmogorov complexity are polynomially 
	classically described states.
\end{example}

As a further interesting special case, we highlight that states that can be prepared by polynomial quantum circuits, even with post-selected measurement results, fall under our definition of classically described states. 

\begin{example}[Quantum circuits with post-selection]\label{ex:circuit}
	Suppose that $\ket\psi$ can be prepared by a quantum circuit of $O(\poly(N))$ gates from $\ket0^{\otimes O(\poly(N))}$, where we allow for post-selected measurement results in the computational basis. Then a Turing machine that classically simulates the circuit constitutes a polynomial classical description in the sense of Definition~\ref{def:classical}. 
\end{example}

\begin{example}[Eigenstates of local Hamiltonians]\label{ex:hamiltonian}
	Suppose that $\ket\psi$ is an {eigenvector} of a local Hamiltonian with bounded interaction {strength}. Such Hamiltonians can be 
	specified to arbitrary (but fixed) precision with polynomial Kolmogorov complexity. Thus, a Turing machine that starts from a polynomial 
	description of the Hamiltonian and computes $\ket\psi$ by brute-force 
	diagonalisation constitutes a polynomial {classical} description of $\ket\psi$ in the sense of 
	Definition~\ref{def:classical}. 
\end{example}

Let us now precisely state what we call an approximation of given pure states by polynomially classically described states.
\begin{definition}[Approximation of quantum many-body states]\label{def:approx}
A family of pure states $\rho_N$ can be \emph{approximated by polynomially classically described states} if for all $\varepsilon>0$, there exist a polynomial $p$ and 
pure states $\omega_N$ with a classical description of length at most $p(N)$ such that for all $N$,
\begin{equation}
	\|\rho_N-\omega_N\|_1\leq \varepsilon.
\end{equation}
\end{definition}
Note that this is exactly the sense in which matrix-product states provide an efficient approximation 
of all one-dimensional states that satisfy an area law for some $S_\alpha$ with $\alpha\in(0,1)$ \cite{SchuchApprox}. 
We remark that Definition~\ref{def:approx} can be weakened without altering the results.

\section{Area laws and approximation by efficiently describable states}
\begin{theorem}[Impossibility of approximating area law states]\label{thm:approx}
	Let $\tilde\hh_L$ be a Hilbert space of dimension $\exp(\Omega(\poly(L)))$. Then there exist states in $\tilde\hh_L$ that cannot be approximated by polynomially classically described states. In particular, not all translationally invariant strong area law states can be approximated by polynomially classically described states.
\end{theorem}
Theorem~\ref{thm:approx} can be easily proven using a counting argument of $\epsilon$-nets. Indeed, the number of states that can be parameterised by 
$O(\poly(L))$ many bits is at most $2^{O(\poly(L))}$. However, an $\epsilon$-net covering the space of pure states in $\cc^q$ requires at least $ (1/\varepsilon)^{\Omega(q)}$ elements  \cite{Nets}, which is much
 larger than $2^{O(\poly(L))}$ if $q=\exp(\Omega(\poly(L)))$ (see also Refs.~\cite{PQSV11,KBGKE11} on the topic of $\varepsilon$-nets for many-body states). We nevertheless also review the  
 more involved proof from Ref.~\cite{Kolmogorov} 
using communication complexity in Appendix~\ref{app:cc}. This proof could,  due to its more constructive nature, provide some insight into the structure of 
some strong area law states that cannot be approximated by polynomially classically described states.

\subsection{Tensor network states}
We saw that our definition of polynomial classical descriptions 
 encompasses all efficient tensor network descriptions. Thus, 

\begin{corollary}[Tensor network states cannot approximate area law states]\label{cor:TN}
	There exist translationally invariant strong area law states that cannot be approximated by polynomial tensor network states in the sense of Example~\ref{ex:tensor-network}. In particular, not all translationally invariant strong area law states can be approximated by polynomial PEPS or MERA states.
\end{corollary}

Notice 
the restriction 
to tensor networks whose tensor entries have a polynomial Kolmogorov complexity. This is required to ensure that the tensor network description is in fact polynomial. Indeed, a classical description depending on only polynomially many parameters $\lambda_1,\ldots,\lambda_{O(\poly(N))}$ (e.g., a PEPS with polynomial bond-dimension) is not necessarily already polynomial -- for the latter, it is also necessary that each of the $\lambda_i$ themselves can be stored efficiently. The notion of Kolmogorov complexity allows for the most general definition of tensor networks that can be stored with polynomial classical memory.

\subsection{Quantum circuits}
Example~\ref{ex:circuit} shows that states prepared by a polynomial quantum circuit with post-selected measurement results have a polynomial classical description. Thus,

\begin{corollary}[Post-selected quantum circuits cannot prepare area law states]\label{cor:circuits}
	There exist translationally invariant strong area law states that cannot be approximated by a polynomial quantum circuit with post-selection in the sense of Example~\ref{ex:circuit}. 
\end{corollary}

In the light of the computational power of post-selected quantum computation \cite{PostBQP}, this may be 
remarkable. 

\subsection{Eigenstates of local Hamiltonians}
Example~\ref{ex:hamiltonian} shows that eigenstates of local Hamiltonians with bounded interaction strengths also have a polynomial classical description. Thus, 

\begin{corollary}[Area law states without parent Hamiltonian]\label{cor:hamiltonian}
	There exist translationally invariant strong area law states that cannot be approximated by eigenstates of local Hamiltonians. 
\end{corollary}

\section{Isotropic states and area laws}
\label{sec:isotropic}
So far,  
the states 
in consideration 
 were translationally invariant but not isotropic. 
However, by taking the superposition of appropriate rotations of \eqref{fullstate}, one can alter the above argument such that all involved states are fully isotropic. 

\begin{theorem}[Isotropic and translationally invariant area law states]\label{thm:isotropic}
There exists an injective linear 
isometry $g:\mathcal G_L\rightarrow \hh$  with $\dim(\mathcal G_L) = \exp(\Omega(L^{D-1}))$
	 such
	that  for all  $\ket{\phi_L}\in \mathcal G_L$, $g(\ket{\phi_L})$ satisfies a strong area law and is isotropic and translationally invariant in all $D$ directions.
\end{theorem}

The details of this construction are given in Appendix~\ref{app:isotropicProof}.

\begin{corollary}[{Approximation for isotropic states}]\label{cor:isotropic}
There exist isotropic and translationally invariant strong area law states that cannot be approximated by polynomially classically described states. In particular, Corollary~\ref{cor:TN}--\ref{cor:hamiltonian} also hold for isotropic and translationally invariant states.
\end{corollary}

\section{Decaying correlations and area laws}
One might wonder whether an exponentially dimensional subspace of strong area law states can be constructed while imposing decaying two-point correlations for distant observables, a property known to occur in ground states of local gapped Hamiltonians \cite{HastingsKoma06}. It follows immediately from their definition
that the states constructed in Theorem~\ref{thm:areaLaw} and \ref{thm:isotropic} already satisfy an algebraic decay
\begin{align}
	\bra{\psi_L}AB\ket{\psi_L} - \bra{\psi_L}A\ket{\psi_L}\!\bra{\psi_L}B\ket{\psi_L} &= O(L^{-1})\nonumber \\
	&=O(\ell^{-1})\label{eqn:corAlgebraic}
\end{align}
for local observables $A,B$ on disjoint supports, where 
$\ell$ is the distance between their supports. Using quantum error-correcting codes, it is however also possible to construct variations of the previous results where all states involved have vanishing two-point correlations of local observables with disjoint support.

To see this, consider a non-degenerate $[n,k,\Delta]$-quantum error-correcting code $C$ with $k/n=\Theta(1)$ and $\Delta/n=\Theta(1)$ \cite{Gottesman96}. Since $C$ is non-degenerate, the reduced density matrix of any $\Delta-1$ qubits of any state in the code space of $C$ is maximally mixed. 
By choosing $n=L^{D-1}$ and considering
\begin{equation} \label{eqn:corSimple}
	\ket{\psi_L} := \ket{C(\chi_L)}\otimes\ket{0}^{(L-1)L^{D-1}},
\end{equation}
where $\ket{C(\chi_L)}$ is an arbitrary state vector in the code space of $C$,
we see that 
\begin{equation} \label{eqn:cor_zero}
	\bra{\psi_L}AB\ket{\psi_L} - \bra{\psi_L}A\ket{\psi_L}\!\bra{\psi_L}B\ket{\psi_L}=0
\end{equation}
for any observables $A,B$ with disjoint support and whose joint support in the top hyperplane is less than $\Delta=\Theta(L^{D-1})$. In particular, Eq.~\eqref{eqn:cor_zero} holds for local observables $A,B$. Clearly, states of the form \eqref{eqn:corSimple} obey a strong area law and since $k=\Theta(L^{D-1})$, we obtain a subspace of dimension $\exp(\Omega(L^{D-1}))$ of strong area law states with vanishing correlations of local observables. 

\begin{corollary}[{Approximation for area law states with vanishing correlations of local observables}]\label{cor:correlations}
There exist strong area law states with vanishing two-point correlations of local observables on disjoint supports that cannot be approximated by polynomially classically described states. In particular, Corollary~\ref{cor:TN}--\ref{cor:hamiltonian} also hold for states with vanishing correlations of local observables on disjoint supports.
\end{corollary}

While the translationally and rotationally invariant construction only gives algebraic decay (Eq.~\eqref{eqn:corAlgebraic}), we conjecture that there also exist strong area law states which are translationally and rotationally invariant and simultaneously have exponentially small correlations for all local observables but still cannot be approximated by polynomially classically described states.

\section{Conclusion and outlook}
In this work, we {have shown} that the set of states satisfying an area law in $D\geq 2$ comprises many states that do not have an efficient classical description: They can be neither approximated by efficient tensor network states, nor using polynomial quantum circuits with post-selected measurements, and are also not eigenstates of local Hamiltonians. 
We have hence proven that the connection between entanglement properties  
and the possibility of an efficient classical description 
is far more intricate than anticipated.
These results are based on the simple observation that an arbitrary quantum state in $D-1$ dimensions that is embedded into $D$ dimensions satisfies a $D$-dimensional area law, implying that the set of area law states contains a subspace of exponential dimension. 
That is to say, in 
$D\geq 2$,
one has the 
freedom to ``dilute'' the entanglement content, in order to still arrive
at area laws. 
We also demonstrated that this exponential scaling persists if various physical properties, such as translational and rotational invariance, or decaying correlations of local observables, are imposed. We note however that whilst the latter can be extended to non-local observables of size $O(L^{D-1})$, our notion of decaying correlations is weaker than the exponential clustering property involving all regions 
valid for ground states of gapped Hamiltonians \cite{HastingsKoma06,BH13}. It remains open whether our results are impeded if the stronger notion of exponential clustering of correlations is imposed.

Area laws indeed suggest the expected correlation patterns of naturally occurring ground states, but when
put in precise contact with questions of numerical simulation, it turns out that satisfying an area law alone is not sufficient
for efficient approximation. Picking up  
the metaphor of the introduction, the
``corner of states that can be efficiently described'' is tiny compared to the ``physical corner'' {(Fig.~\ref{fig:corner})}.

\begin{figure}[t]
	\subfloat[][]{
		\includegraphics[width=.23\textwidth]{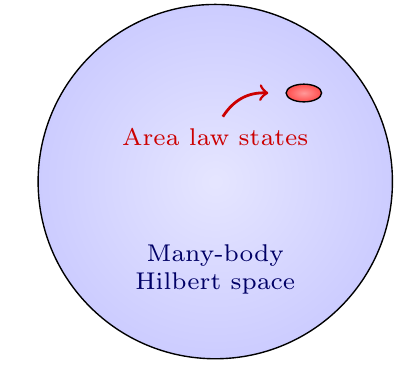}
		\label{subfig:corner1}
	}
	\subfloat[][]{
		\includegraphics[width=.23\textwidth]{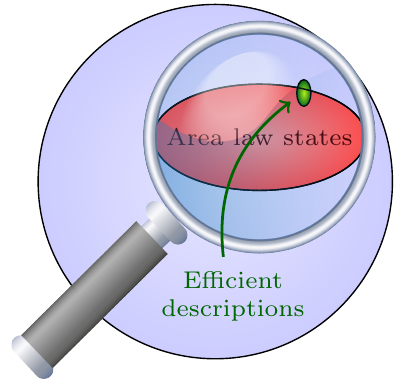}
		\label{subfig:corner2}
	}
	\caption{(a) The set of area law states is a tiny "corner" of the many-body Hilbert space. (b) The set of states that can be efficiently described is tiny compared to the "corner" of area law states.}
	\label{fig:corner}
\end{figure}

The construction using 
communication complexity exhibits a semi-explicit class of area law states 
without
a short classical description. This can be taken as a starting point for further investigation with the aim to identify additional criteria of ``physical'' states whose imposition supplementary to area laws could reduce the ``{physical} corner'' to sub-exponential 
{size}. 

A particularly exciting perspective 
arises 
from the observation that states with  
small 
entanglement content can go along with states having divergent bond dimensions in PEPS approximations. 
This may be taken as a suggestion that there may be states that are in the same phase if symmetries are imposed, but 
are being classified as being in different phases in a classification of phases of matter 
building upon tensor network descriptions \cite{ClassifyingPhasesWen,ClassifyingPhasesPollmann,ClassifyingPhasesSchuch}.
It is the hope that the present work can be taken as a starting point of further endeavours towards
understanding the complexity of quantum many-body states.

\acknowledgements
	We thank J.~I.~Cirac,  A.~Ferris, M.~Friesdorf, C.~Gogolin, Y.-K.~Liu, A.~Moln\'{a}r, X.~Ni, N.~Schuch, and H.~Wilming for helpful discussions.
 We acknowledge funding from the BMBF, the 
 EU (RAQUEL, SIQS, COST, AQuS), and the ERC (TAQ).

\appendix

\section{Proof of Theorem~\ref{thm:approx} using communication complexity}\label{app:cc}

Suppose two distant parties, Alice and Bob, each possess an $n$-bit string, $\xx$ and $\yy$, respectively. No communication between Alice and Bob is 
allowed, but they can communicate with a third party, Charlie, whose task 
is to guess whether or not $\xx = \yy$. We demand that Charlie may guess the wrong answer with a small (fixed) probability of at most $\delta>0$. This is called the \emph{equality problem}, which we denote by $\Eq(n)$. 
We now state some known results \cite{Kolmogorov,NewmanSzegedy96,QuantumFingerprinting} on the communication complexity, i.e. the minimum amount of communication required 
for solving the equality problem.

\begin{lemma}[Equality problem for classical communication]\label{lem:cc-classical}
	If Alice and Bob can only send classical information to Charlie, at least $\Omega(\sqrt n)$ bits of communication are required to solve $\Eq(n)$.
\end{lemma}

\begin{lemma}[Quantum solution to equality problem]\label{lem:cc-quantum}
\ 
	\begin{enumerate}[(i)]
	\item \label{lem:cc-quantum-protocol}
	If Alice and Bob can send quantum information to Charlie, there exists a protocol for $\Eq(n)$ using only $O(\log n)$ qubits of communication that is of the following form: 
	Alice and Bob each prepare $\ket{h(\xx)}$ and $\ket{h(\yy)}$ of $O(\log n)$-qubits, respectively, 
	which they send to Charlie. Charlie then applies a quantum circuit 
	to $\ket{h(\xx)}\ket{h(\yy)}\ket 0$, followed by a measurement of a single qubit whose outcome determines Charlie's guess. 
	\item\label{lem:cc-quantum-epsilon}
	There exists an $\varepsilon>0$ independent of $n$ such that the protocol in (\ref{lem:cc-quantum-protocol}) still works if instead,  Alice and Bob send states to Charlie which are $\varepsilon$-close  in trace distance\footnote{This was argued in Ref.\ \cite{Kolmogorov} for the Euclidean vector distance but it is clear that the same holds for the trace distance.} to $\ket{h(\xx)}$ and $\ket{h(\yy)}$.
	\end{enumerate}
\end{lemma}

We now turn to the proof of Theorem~\ref{thm:approx}.

\begin{proof}[Proof of Theorem~\ref{thm:approx}]
We prove the claim by contradiction. Suppose that every state vector in $\tilde\hh_L$ can be approximated by polynomially classically described states. Then in particular, all $M$-qubit states can be approximated by states with a classical description of length $O(\poly(N))$, where $M:=\lfloor\log_2\dim(\tilde\hh_L)\rfloor$. 
Fix $\delta\in(0,1)$ and let $\varepsilon>0$ be as in Lemma~\ref{lem:cc-quantum} (\ref{lem:cc-quantum-epsilon}). By Lemma~\ref{lem:cc-quantum} (\ref{lem:cc-quantum-protocol}), we can choose $n$ with $\log n = \Theta(M)$ such that $M$ qubits of communication suffice to solve $\Eq(n)$. 

By assumption, $\ket{h(\xx)}$ and $\ket{h(\yy)}$ can be $\varepsilon$-approximated by states which have {an} 
$O(\poly(M))$ classical description. By Lemma~\ref{lem:cc-quantum} (\ref{lem:cc-quantum-epsilon}), these states can be used instead of $\ket{h(\xx)}$ and $\ket{h(\yy)}$ in the quantum protocol to solve $\Eq(n)$. 
Now consider an alternative protocol using only classical communication to solve $\Eq(n)$ as follows: Alice and Bob send the classical description of their states to Charlie, who simulates the quantum circuit and the measurement  from Lemma~\ref{lem:cc-quantum} using the classical descriptions of the states. This protocol solves $\Eq(n)$ using only $O(\poly(M))=O(\poly(\log n))$ bits of communication, contradicting Lemma~\ref{lem:cc-classical}. 
Finally, by setting $\tilde\hh_L := f(\hh_L)$ with $f$ and $\hh_L$ as in Theorem~\ref{thm:areaLaw}, 
the second part of Theorem~\ref{thm:approx} follows.
\end{proof}

\section{Proof of Theorem~\ref{thm:isotropic}} \label{app:isotropicProof}

Theorem~\ref{thm:isotropic} can be proven with a minor modification of the proof of Theorem~\ref{thm:areaLaw}. 
To start with, we replace $\ket{\phi_L}$ for each $L$ by a mirror
symmetric state vector on the translationally invariant subset $\hh_L\subset (\cc^2)^{\otimes L^{D-1}}$. 
We then consider for the entire $[L]^D$ lattice state vectors of the form
\begin{equation}
	\ket{\Psi_L} := D^{-1/2} \sum_{j=1}^D {\cal R}_j \ket{\psi_L} ,
\end{equation}
where $\id = {\cal R}_1,\ldots, {\cal R}_D$ rotate the entire lattice system such that $|\phi_L\rangle$ is arranged along each line of the cubic lattice in dimension $D$.
Such a state is translationally invariant and isotropic, following from mirror symmetry. These states satisfy a strong area law: For any cubic subset $A\subset [L]^D$,
\begin{equation}
	(\psi_L)_A = D^{-1} \sum_{j=1}^D \tr_{\bar{A}}({\cal R}_j |\psi_L\rangle \langle \psi_L| {\cal R}_j^\dagger),
\end{equation}
since for $j\neq k$,
\begin{equation}
	 \tr_{\bar{A}}({\cal R}_j |\psi_L\rangle \langle \psi_L| {\cal R}_k^\dagger)=0.
\end{equation}
This can be seen by taking the partial trace with respect to a set $C$ first. For simplicity of notation, 
for $D=2$,
consider w.l.o.g.\ distinguished subsets $A\subset [L]^D$ for which $A\cap C=\emptyset$ for $C:= L\times [L]$.
Then, 
\begin{align}
	 \tr_{\bar{A}}( |\psi_L\rangle \langle \psi_L| {\cal R}_2^\dagger)&= \tr_{{\bar{A}\backslash C }}
	 \sum_{\xx\in S}\langle \xx | \psi_L\rangle\langle \psi_L| {\cal R}_2^\dagger|\xx\rangle\nonumber
	 \\ &
	 =0,
\end{align}
where $S=\{\xx|\,\exists j: x_j\neq 0 \wedge x_k=0\, \forall k\in [L]\backslash\{ j\} \}$. An analogous argument holds for any dimension $D$. From these considerations, 
it follows that the area law is inherited by the 
area law valid for each individual ${\cal R}_j |\psi_L\rangle$.
It is furthermore clear that the exponential scaling of the dimension is not affected by restricting to the subspace $\mathcal G_L\subset \hh_L$ of mirror symmetric states.
\qed


\begin{thebibliography}{32}
\expandafter\ifx\csname natexlab\endcsname\relax\def\natexlab#1{#1}\fi
\expandafter\ifx\csname bibnamefont\endcsname\relax
  \def\bibnamefont#1{#1}\fi
\expandafter\ifx\csname bibfnamefont\endcsname\relax
  \def\bibfnamefont#1{#1}\fi
\expandafter\ifx\csname citenamefont\endcsname\relax
  \def\citenamefont#1{#1}\fi
\expandafter\ifx\csname url\endcsname\relax
  \def\url#1{\texttt{#1}}\fi
\expandafter\ifx\csname urlprefix\endcsname\relax\def\urlprefix{URL }\fi
\providecommand{\bibinfo}[2]{#2}
\providecommand{\eprint}[2][]{\url{#2}}

\bibitem[{\citenamefont{Eisert et~al.}(2010)\citenamefont{Eisert, Cramer, and
  Plenio}}]{AreaRMP}
\bibinfo{author}{\bibfnamefont{J.}~\bibnamefont{Eisert}},
  \bibinfo{author}{\bibfnamefont{M.}~\bibnamefont{Cramer}}, \bibnamefont{and}
  \bibinfo{author}{\bibfnamefont{M.~B.} \bibnamefont{Plenio}},
  \emph{\bibinfo{title}{Area laws for the entanglement entropy}},
  \bibinfo{journal}{Rev. Mod. Phys.} \textbf{\bibinfo{volume}{82}},
  \bibinfo{pages}{277} (\bibinfo{year}{2010}).

\bibitem[{\citenamefont{Hastings}(2007{\natexlab{a}})}]{hastings07}
\bibinfo{author}{\bibfnamefont{M.~B.} \bibnamefont{Hastings}},
  \emph{\bibinfo{title}{{An area law for one-dimensional quantum systems}}},
  \bibinfo{journal}{J. Stat. Mech} \textbf{\bibinfo{volume}{P08024}}
  (\bibinfo{year}{2007}{\natexlab{a}}).

\bibitem[{\citenamefont{Arad et~al.}(2012)\citenamefont{Arad, Landau, and
  Vazirani}}]{ALV12}
\bibinfo{author}{\bibfnamefont{I.}~\bibnamefont{Arad}},
  \bibinfo{author}{\bibfnamefont{Z.}~\bibnamefont{Landau}}, \bibnamefont{and}
  \bibinfo{author}{\bibfnamefont{U.}~\bibnamefont{Vazirani}},
  \emph{\bibinfo{title}{Improved one-dimensional area law for frustration-free
  systems}}, \bibinfo{journal}{Phys. Rev. B} \textbf{\bibinfo{volume}{85}},
  \bibinfo{pages}{195145} (\bibinfo{year}{2012}).

\bibitem[{\citenamefont{Arad et~al.}()\citenamefont{Arad, Kitaev, Landau, and
  Vazirani}}]{AKLV13}
\bibinfo{author}{\bibfnamefont{I.}~\bibnamefont{Arad}},
  \bibinfo{author}{\bibfnamefont{A.}~\bibnamefont{Kitaev}},
  \bibinfo{author}{\bibfnamefont{Z.}~\bibnamefont{Landau}}, \bibnamefont{and}
  \bibinfo{author}{\bibfnamefont{U.}~\bibnamefont{Vazirani}},
  \emph{\bibinfo{title}{An area law and sub-exponential algorithm for 1d
  systems}}, \bibinfo{note}{arXiv:1301.1162}.

\bibitem[{\citenamefont{Brand{\~a}o and Horodecki}(2013)}]{BH13}
\bibinfo{author}{\bibfnamefont{F.~G. S.~L.} \bibnamefont{Brand{\~a}o}}
  \bibnamefont{and}
  \bibinfo{author}{\bibfnamefont{M.}~\bibnamefont{Horodecki}},
  \emph{\bibinfo{title}{An area law for entanglement from exponential decay of
  correlations}}, \bibinfo{journal}{Nature Physics}
  \textbf{\bibinfo{volume}{9}}, \bibinfo{pages}{721} (\bibinfo{year}{2013}).

\bibitem[{\citenamefont{Plenio et~al.}(2005)\citenamefont{Plenio, Eisert,
  Dreissig, and Cramer}}]{Area}
\bibinfo{author}{\bibfnamefont{M.~B.} \bibnamefont{Plenio}},
  \bibinfo{author}{\bibfnamefont{J.}~\bibnamefont{Eisert}},
  \bibinfo{author}{\bibfnamefont{J.}~\bibnamefont{Dreissig}}, \bibnamefont{and}
  \bibinfo{author}{\bibfnamefont{M.}~\bibnamefont{Cramer}},
  \emph{\bibinfo{title}{Entropy, entanglement, and area: analytical results for
  harmonic lattice systems}}, \bibinfo{journal}{Phys. Rev. Lett.}
  \textbf{\bibinfo{volume}{94}}, \bibinfo{pages}{060503}
  (\bibinfo{year}{2005}).

\bibitem[{\citenamefont{Cramer and Eisert}(2006)}]{Area2}
\bibinfo{author}{\bibfnamefont{M.}~\bibnamefont{Cramer}} \bibnamefont{and}
  \bibinfo{author}{\bibfnamefont{J.}~\bibnamefont{Eisert}},
  \emph{\bibinfo{title}{Correlations, spectral gap, and entanglement in
  harmonic quantum systems on generic lattices}}, \bibinfo{journal}{New J.
  Phys.} \textbf{\bibinfo{volume}{8}}, \bibinfo{pages}{71}
  (\bibinfo{year}{2006}).

\bibitem[{\citenamefont{Cramer et~al.}(2006)\citenamefont{Cramer, Eisert,
  Plenio, and Dreissig}}]{Area3}
\bibinfo{author}{\bibfnamefont{M.}~\bibnamefont{Cramer}},
  \bibinfo{author}{\bibfnamefont{J.}~\bibnamefont{Eisert}},
  \bibinfo{author}{\bibfnamefont{M.~B.} \bibnamefont{Plenio}},
  \bibnamefont{and} \bibinfo{author}{\bibfnamefont{J.}~\bibnamefont{Dreissig}},
  \emph{\bibinfo{title}{An entanglement-area law for general bosonic harmonic
  lattice systems}}, \bibinfo{journal}{Phys. Rev. A}
  \textbf{\bibinfo{volume}{73}}, \bibinfo{pages}{012309}
  (\bibinfo{year}{2006}).

\bibitem[{\citenamefont{Van~Acoleyen et~al.}(2013)\citenamefont{Van~Acoleyen,
  Mari{\"e}n, and Verstraete}}]{vAMV13}
\bibinfo{author}{\bibfnamefont{K.}~\bibnamefont{Van~Acoleyen}},
  \bibinfo{author}{\bibfnamefont{M.}~\bibnamefont{Mari{\"e}n}},
  \bibnamefont{and}
  \bibinfo{author}{\bibfnamefont{F.}~\bibnamefont{Verstraete}},
  \emph{\bibinfo{title}{Entanglement rates and area laws}},
  \bibinfo{journal}{Phys. Rev. Lett.} \textbf{\bibinfo{volume}{111}},
  \bibinfo{pages}{170501} (\bibinfo{year}{2013}).

\bibitem[{\citenamefont{Mari\"{e}n et~al.}()\citenamefont{Mari\"{e}n,
  Audenaert, Acoleyen, and Verstraete}}]{MAAV14}
\bibinfo{author}{\bibfnamefont{M.}~\bibnamefont{Mari\"{e}n}},
  \bibinfo{author}{\bibfnamefont{K.~M.} \bibnamefont{Audenaert}},
  \bibinfo{author}{\bibfnamefont{K.~V.} \bibnamefont{Acoleyen}},
  \bibnamefont{and}
  \bibinfo{author}{\bibfnamefont{F.}~\bibnamefont{Verstraete}},
  \emph{\bibinfo{title}{Entanglement rates and the stability of the area law
  for the entanglement entropy}}, \bibinfo{note}{arXiv:1411.0680}.

\bibitem[{\citenamefont{Brandao and Cramer}()}]{HeatCapacity}
\bibinfo{author}{\bibfnamefont{F.~G. S.~L.} \bibnamefont{Brandao}}
  \bibnamefont{and} \bibinfo{author}{\bibfnamefont{M.}~\bibnamefont{Cramer}},
  \emph{\bibinfo{title}{Entanglement area law from specific heat capacity}},
  \bibinfo{note}{arXiv:1409.5946}.

\bibitem[{\citenamefont{Hastings}(2007{\natexlab{b}})}]{HastingsEntropyEntanglement}
\bibinfo{author}{\bibfnamefont{M.~B.} \bibnamefont{Hastings}},
  \emph{\bibinfo{title}{Entropy and entanglement in quantum ground states}},
  \bibinfo{journal}{Phys. Rev. B} \textbf{\bibinfo{volume}{76}},
  \bibinfo{pages}{035114} (\bibinfo{year}{2007}{\natexlab{b}}).

\bibitem[{\citenamefont{Masanes}(2009)}]{MasanesArea}
\bibinfo{author}{\bibfnamefont{L.}~\bibnamefont{Masanes}},
  \emph{\bibinfo{title}{Area law for the entropy of low-energy states}},
  \bibinfo{journal}{Phys. Rev. A} \textbf{\bibinfo{volume}{80}},
  \bibinfo{pages}{052104} (\bibinfo{year}{2009}).

\bibitem[{\citenamefont{de~Beaudrap et~al.}(2010)\citenamefont{de~Beaudrap,
  Ohliger, Osborne, and Eisert}}]{FF}
\bibinfo{author}{\bibfnamefont{N.}~\bibnamefont{de~Beaudrap}},
  \bibinfo{author}{\bibfnamefont{M.}~\bibnamefont{Ohliger}},
  \bibinfo{author}{\bibfnamefont{T.~J.} \bibnamefont{Osborne}},
  \bibnamefont{and} \bibinfo{author}{\bibfnamefont{J.}~\bibnamefont{Eisert}},
  \emph{\bibinfo{title}{Solving frustration-free spin systems}},
  \bibinfo{journal}{Phys. Rev. Lett.} \textbf{\bibinfo{volume}{105}},
  \bibinfo{pages}{060504} (\bibinfo{year}{2010}).

\bibitem[{\citenamefont{Michalakis}()}]{michalakis2012}
\bibinfo{author}{\bibfnamefont{S.}~\bibnamefont{Michalakis}},
  \emph{\bibinfo{title}{{Stability of the area law for the entropy of
  entanglement}}}, \bibinfo{note}{arXiv:1206.6900}.

\bibitem[{\citenamefont{Schollw{\"o}ck}(2011)}]{schollwoeck2011}
\bibinfo{author}{\bibfnamefont{U.}~\bibnamefont{Schollw{\"o}ck}},
  \emph{\bibinfo{title}{The density-matrix renormalization group in the age of
  matrix product states}}, \bibinfo{journal}{Ann. Phys.}
  \textbf{\bibinfo{volume}{326}}, \bibinfo{pages}{96} (\bibinfo{year}{2011}).

\bibitem[{\citenamefont{Verstraete et~al.}(2008)\citenamefont{Verstraete, Murg,
  and Cirac}}]{VMC08}
\bibinfo{author}{\bibfnamefont{F.}~\bibnamefont{Verstraete}},
  \bibinfo{author}{\bibfnamefont{V.}~\bibnamefont{Murg}}, \bibnamefont{and}
  \bibinfo{author}{\bibfnamefont{J.~I.} \bibnamefont{Cirac}},
  \emph{\bibinfo{title}{Matrix product states, projected entangled pair states,
  and variational renormalization group methods for quantum spin systems}},
  \bibinfo{journal}{Adv. Phys.} \textbf{\bibinfo{volume}{57}},
  \bibinfo{pages}{143} (\bibinfo{year}{2008}).

\bibitem[{\citenamefont{Schuch et~al.}(2008)\citenamefont{Schuch, Wolf,
  Verstraete, and Cirac}}]{SchuchApprox}
\bibinfo{author}{\bibfnamefont{N.}~\bibnamefont{Schuch}},
  \bibinfo{author}{\bibfnamefont{M.~M.} \bibnamefont{Wolf}},
  \bibinfo{author}{\bibfnamefont{F.}~\bibnamefont{Verstraete}},
  \bibnamefont{and} \bibinfo{author}{\bibfnamefont{J.~I.} \bibnamefont{Cirac}},
  \emph{\bibinfo{title}{Entropy scaling and simulability by matrix product
  states}}, \bibinfo{journal}{Phys. Rev. Lett.} \textbf{\bibinfo{volume}{100}},
  \bibinfo{pages}{030504} (\bibinfo{year}{2008}).

\bibitem[{\citenamefont{Verstraete and Cirac}(2006)}]{VC06}
\bibinfo{author}{\bibfnamefont{F.}~\bibnamefont{Verstraete}} \bibnamefont{and}
  \bibinfo{author}{\bibfnamefont{J.~I.} \bibnamefont{Cirac}},
  \emph{\bibinfo{title}{{Matrix product states represent ground states
  faithfully}}}, \bibinfo{journal}{Phys. Rev. B} \textbf{\bibinfo{volume}{73}},
  \bibinfo{pages}{94423} (\bibinfo{year}{2006}).

\bibitem[{\citenamefont{Vidal}(2007)}]{MERA1}
\bibinfo{author}{\bibfnamefont{G.}~\bibnamefont{Vidal}},
  \emph{\bibinfo{title}{Entanglement renormalization}}, \bibinfo{journal}{Phys.
  Rev. Lett.} \textbf{\bibinfo{volume}{99}}, \bibinfo{pages}{220405}
  (\bibinfo{year}{2007}).

\bibitem[{\citenamefont{Mora et~al.}(2007)\citenamefont{Mora, Briegel, and
  Kraus}}]{Kolmogorov}
\bibinfo{author}{\bibfnamefont{C.}~\bibnamefont{Mora}},
  \bibinfo{author}{\bibfnamefont{H.}~\bibnamefont{Briegel}}, \bibnamefont{and}
  \bibinfo{author}{\bibfnamefont{B.}~\bibnamefont{Kraus}},
  \emph{\bibinfo{title}{{Quantum Kolmogorov complexity and its applications}}},
  \bibinfo{journal}{Int. J. Quant. Inf.} \textbf{\bibinfo{volume}{05}},
  \bibinfo{pages}{729} (\bibinfo{year}{2007}).

\bibitem[{\citenamefont{Hayden}(2010)}]{Nets}
\bibinfo{author}{\bibfnamefont{P.}~\bibnamefont{Hayden}},
  \emph{\bibinfo{title}{Concentration of measure effects in quantum
  information}}, \bibinfo{journal}{Proc. Symp. App. Math.}
  \textbf{\bibinfo{volume}{68}} (\bibinfo{year}{2010}).

\bibitem[{\citenamefont{Poulin et~al.}(2011)\citenamefont{Poulin, Qarry, Somma,
  and Verstraete}}]{PQSV11}
\bibinfo{author}{\bibfnamefont{D.}~\bibnamefont{Poulin}},
  \bibinfo{author}{\bibfnamefont{A.}~\bibnamefont{Qarry}},
  \bibinfo{author}{\bibfnamefont{R.}~\bibnamefont{Somma}}, \bibnamefont{and}
  \bibinfo{author}{\bibfnamefont{F.}~\bibnamefont{Verstraete}},
  \emph{\bibinfo{title}{{Quantum simulation of time-dependent Hamiltonians and
  the convenient illusion of Hilbert space}}}, \bibinfo{journal}{Phys. Rev.
  Lett.} \textbf{\bibinfo{volume}{106}}, \bibinfo{pages}{170501}
  (\bibinfo{year}{2011}).

\bibitem[{\citenamefont{Kliesch et~al.}(2011)\citenamefont{Kliesch, Barthel,
  Gogolin, Kastoryano, and Eisert}}]{KBGKE11}
\bibinfo{author}{\bibfnamefont{M.}~\bibnamefont{Kliesch}},
  \bibinfo{author}{\bibfnamefont{T.}~\bibnamefont{Barthel}},
  \bibinfo{author}{\bibfnamefont{C.}~\bibnamefont{Gogolin}},
  \bibinfo{author}{\bibfnamefont{M.}~\bibnamefont{Kastoryano}},
  \bibnamefont{and} \bibinfo{author}{\bibfnamefont{J.}~\bibnamefont{Eisert}},
  \emph{\bibinfo{title}{{Dissipative Quantum Church-Turing theorem}}},
  \bibinfo{journal}{Phys. Rev. Lett.} \textbf{\bibinfo{volume}{107}},
  \bibinfo{pages}{120501} (\bibinfo{year}{2011}).

\bibitem[{\citenamefont{Aaronson}(2005)}]{PostBQP}
\bibinfo{author}{\bibfnamefont{S.}~\bibnamefont{Aaronson}},
  \emph{\bibinfo{title}{Quantum computing, post-selection, and probabilistic
  polynomial-time}}, \bibinfo{journal}{Proc. Roy. Soc. A}
  \textbf{\bibinfo{volume}{461}}, \bibinfo{pages}{3473} (\bibinfo{year}{2005}).

\bibitem[{\citenamefont{Hastings and Koma}(2006)}]{HastingsKoma06}
\bibinfo{author}{\bibfnamefont{M.~B.} \bibnamefont{Hastings}} \bibnamefont{and}
  \bibinfo{author}{\bibfnamefont{T.}~\bibnamefont{Koma}},
  \emph{\bibinfo{title}{Spectral gap and exponential decay of correlations}},
  \bibinfo{journal}{Commun. Math. Phys.} \textbf{\bibinfo{volume}{265}},
  \bibinfo{pages}{781} (\bibinfo{year}{2006}).

\bibitem[{\citenamefont{Gottesman}(1996)}]{Gottesman96}
\bibinfo{author}{\bibfnamefont{D.}~\bibnamefont{Gottesman}},
  \emph{\bibinfo{title}{{Class of quantum error-correcting codes saturating the
  quantum Hamming bound}}}, \bibinfo{journal}{Phys. Rev. A}
  \textbf{\bibinfo{volume}{54}}, \bibinfo{pages}{1862} (\bibinfo{year}{1996}).

\bibitem[{\citenamefont{Chen et~al.}(2011)\citenamefont{Chen, Gu, and
  Wen}}]{ClassifyingPhasesWen}
\bibinfo{author}{\bibfnamefont{X.}~\bibnamefont{Chen}},
  \bibinfo{author}{\bibfnamefont{Z.-C.} \bibnamefont{Gu}}, \bibnamefont{and}
  \bibinfo{author}{\bibfnamefont{X.-G.} \bibnamefont{Wen}},
  \emph{\bibinfo{title}{Complete classification of 1d gapped quantum phases in
  interacting spin systems}}, \bibinfo{journal}{Phys. Rev. B}
  \textbf{\bibinfo{volume}{83}}, \bibinfo{pages}{035107}
  (\bibinfo{year}{2011}).

\bibitem[{\citenamefont{Turner et~al.}(2011)\citenamefont{Turner, Pollmann, and
  Berg}}]{ClassifyingPhasesPollmann}
\bibinfo{author}{\bibfnamefont{A.~M.} \bibnamefont{Turner}},
  \bibinfo{author}{\bibfnamefont{F.}~\bibnamefont{Pollmann}}, \bibnamefont{and}
  \bibinfo{author}{\bibfnamefont{E.}~\bibnamefont{Berg}},
  \emph{\bibinfo{title}{Topological phases of one-dimensional fermions: An
  entanglement point of view}}, \bibinfo{journal}{Phys. Rev. B}
  \textbf{\bibinfo{volume}{83}}, \bibinfo{pages}{075102}
  (\bibinfo{year}{2011}).

\bibitem[{\citenamefont{Schuch et~al.}(2011)\citenamefont{Schuch, Perez-Garcia,
  and Cirac}}]{ClassifyingPhasesSchuch}
\bibinfo{author}{\bibfnamefont{N.}~\bibnamefont{Schuch}},
  \bibinfo{author}{\bibfnamefont{D.}~\bibnamefont{Perez-Garcia}},
  \bibnamefont{and} \bibinfo{author}{\bibfnamefont{I.}~\bibnamefont{Cirac}},
  \emph{\bibinfo{title}{{Classifying quantum phases using matrix product states
  and PEPS}}}, \bibinfo{journal}{Phys. Rev. B} \textbf{\bibinfo{volume}{84}},
  \bibinfo{pages}{165139} (\bibinfo{year}{2011}).

\bibitem[{\citenamefont{Newman and Szegedy}(1996)}]{NewmanSzegedy96}
\bibinfo{author}{\bibfnamefont{I.}~\bibnamefont{Newman}} \bibnamefont{and}
  \bibinfo{author}{\bibfnamefont{M.}~\bibnamefont{Szegedy}},
  \emph{\bibinfo{title}{Public vs. private coin flips in one round
  communication games (extended abstract)}}, in \emph{\bibinfo{booktitle}{In
  Proc. 28th ACM Symp. on the Theory of Computing}} (\bibinfo{publisher}{ACM
  Press}, \bibinfo{year}{1996}), pp. \bibinfo{pages}{561--570}.

\bibitem[{\citenamefont{Buhrman et~al.}(2001)\citenamefont{Buhrman, Cleve,
  Watrous, and de~Wolf}}]{QuantumFingerprinting}
\bibinfo{author}{\bibfnamefont{H.}~\bibnamefont{Buhrman}},
  \bibinfo{author}{\bibfnamefont{R.}~\bibnamefont{Cleve}},
  \bibinfo{author}{\bibfnamefont{J.}~\bibnamefont{Watrous}}, \bibnamefont{and}
  \bibinfo{author}{\bibfnamefont{R.}~\bibnamefont{de~Wolf}},
  \emph{\bibinfo{title}{Quantum fingerprinting}}, \bibinfo{journal}{Phys. Rev.
  Lett.} \textbf{\bibinfo{volume}{87}}, \bibinfo{pages}{167902}
  (\bibinfo{year}{2001}).

\end{thebibliography}
\end{document}